%% file: main-journal.tex
\newif\ifstoc
\stoctrue 
\stocfalse
\documentclass[11pt]{article}
\usepackage{hyperref}
\hypersetup{colorlinks=true,
            citebordercolor={.6 .6 .6},linkbordercolor={.6 .6 .6},%
citecolor=blue,urlcolor=black,linkcolor=blue,pagecolor=black}

\usepackage{amsmath,amssymb,bbm,amsthm}
\usepackage{fullpage}
\usepackage{thm-restate,color,xcolor,xspace}
\usepackage[capitalise,nameinlink]{cleveref}
\usepackage{graphicx}


\usepackage{algorithm}
\usepackage{algpseudocode}
\usepackage{algorithmicx}
\usepackage{setspace}
\usepackage[compact]{titlesec}
\allowdisplaybreaks

\newcommand{\cd}{\cdot}
\newcommand{\bn}{\binom}

\newcommand{\lds}{\ldots}

\newcommand{\sm}{\setminus}
\newcommand{\s}{\subseteq}

\newcommand{\BE}{\begin{enumerate}}
\newcommand{\EE}{\end{enumerate}}

\newcommand{\BI}{\begin{itemize}}
\newcommand{\EI}{\end{itemize}}

\newcommand{\eps}{\varepsilon}
\renewcommand{\epsilon}{\varepsilon}

\newcommand{\la}{\lambda}

\newcommand{\al}{\alpha}

\newcommand{\Om}{\Omega}
\newcommand{\el}{\ell}

\newcommand{\m}{\mathcal}

\newcommand{\lc}{\lceil}
\newcommand{\rc}{\rceil}

\newcommand{\poly}{\mathrm{poly}}
\newcommand{\polylog}{\mathrm{polylog}}

\newcommand{\lmt}{\left[\begin{matrix}}
\newcommand{\rmt}{\end{matrix}\right]}

\newtheorem{theorem}{Theorem}
\newtheorem{lemma}[theorem]{Lemma}
\newtheorem{definition}[theorem]{Definition}
\newtheorem{corollary}[theorem]{Corollary}
\newtheorem{observation}[theorem]{Observation}
\newtheorem{proposition}[theorem]{Proposition}
\newtheorem{claim}[theorem]{Claim}
\newtheorem{subclaim}{Subclaim}
\newtheorem{fact}[theorem]{Fact}
\newtheorem{assumption}[theorem]{Assumption}
\newcommand{\BT}{\begin{theorem}}
\newcommand{\ET}{\end{theorem}}
\newcommand{\BL}{\begin{lemma}}
\newcommand{\EL}{\end{lemma}}
\newcommand{\BD}{\begin{definition}}
\newcommand{\ED}{\end{definition}}
\newcommand{\BC}{\begin{corollary}}
\newcommand{\EC}{\end{corollary}}
\newcommand{\BO}{\begin{observation}}
\newcommand{\EO}{\end{observation}}
\newcommand{\BCL}{\begin{claim}}
\newcommand{\ECL}{\end{claim}}
\newcommand{\BSCL}{\begin{subclaim}}
\newcommand{\ESCL}{\end{subclaim}}
\newcommand{\BF}{\begin{fact}}
\newcommand{\EF}{\end{fact}}
\newcommand{\BA}{\begin{assumption}}
\newcommand{\EA}{\end{assumption}}
\newcommand{\BP}{\begin{proof}}
\newcommand{\EP}{\end{proof}}
\newcommand{\BPS}{\begin{proof}[Proof (Sketch)]}
\newcommand{\EPS}{\end{proof}}
\Crefname{observation}{Observation}{Observations}
\Crefname{claim}{Claim}{Claims}
\Crefname{subclaim}{Subclaim}{Subclaims}
\Crefname{fact}{Fact}{Facts}
\Crefname{assumption}{Assumption}{Assumptions}

\newcommand{\tO}{\tilde{O}}

\newcommand{\leml}[1]{\label{lem:#1}}
\newcommand{\lem}[1]{\Cref{lem:#1}}

\newcommand{\Venn}{\textup{Venn}}

\newcommand{\lk}{\overline\lambda_k}

\newcommand{\kcut}{\ensuremath{k\textsc{-Cut}}\xspace}
 
\newcommand{\kclique}{\textsc{Max-Weight~\ensuremath{k}\textsc{-Clique}}\xspace}
\newcommand{\kkclique}{\textsc{Max-Weight~\ensuremath{(k-1)}\textsc{-Clique}}\xspace}

\newcommand{\sun}{\mathsf{sf}}

\newcommand{\bE}{\ensuremath{\mathbf{E}}}

\newcommand{\initOneLiners}{%
    \setlength{\itemsep}{0pt}
    \setlength{\parsep }{0pt}
    \setlength{\topsep }{0pt}
}

\begin{document}

\title{\textbf{Optimal Bounds for the $k$-cut Problem}}
\author{ Anupam Gupta\thanks{{\tt anupamg@cs.cmu.edu}. Supported in part by NSF award CCF-1907820, CCF-1955785, and CCF-2006953, and the Indo-US Joint Center for Algorithms Under Uncertainty. } \\ CMU
\and David G. Harris\thanks{{\tt davidgharris29@gmail.com}. } \\ UMD
 \and Euiwoong Lee\thanks{{\tt euiwoong@umich.edu}. Supported in part by the Simons Collaboration on Algorithms and Geometry. }\\ University of Michigan
  \and Jason Li\thanks{{\tt jmli@cs.cmu.edu}. Supported in part by NSF award
    CCF-1907820. }\\ CMU}
\date{}

\maketitle
\thispagestyle{empty}

\begin{abstract}
  In the $k$-cut problem, we want
  to find the lowest-weight set of edges whose deletion breaks a given 
  (multi)graph into $k$ connected components. Algorithms of Karger \& Stein
  can solve this in roughly $O(n^{2k})$ time. On the other hand, lower bounds from
  conjectures about the $k$-clique problem imply that $\Omega(n^{(1-o(1))k})$ time is likely needed.  Recent results of Gupta, Lee \& Li have
  given new algorithms for general $k$-cut in 
  $n^{1.98k + O(1)}$ time, as well as specialized algorithms with better performance for
  certain classes of graphs (e.g., for small integer edge weights).

  In this work, we resolve the problem for general graphs. We show  that  the Contraction Algorithm of Karger outputs any fixed $k$-cut of weight $\alpha \lambda_k$
  with probability $\Omega_k(n^{-\alpha k})$, where $\lambda_k$ denotes the minimum $k$-cut weight.  This also gives an extremal bound of $O_k(n^k)$
  on the number of minimum $k$-cuts    and an algorithm to compute $\lambda_k$ with  roughly $n^k \polylog(n)$ runtime. 
  Both are tight up to lower-order factors, with the algorithmic lower bound assuming hardness of max-weight $k$-clique.

  The first main ingredient in our result is an extremal bound on the number of cuts of weight less than
  $2 \lambda_k/k$, using the Sunflower lemma. The second ingredient is a fine-grained analysis
  of how the graph shrinks -- and how the average degree
  evolves -- in the Karger process.
  
\end{abstract}

\newpage

\setcounter{page}{1}

\input{intro}

\input{techniques}
\input{relnote}

\appendix
\input{appendix}

{\small
  \bibliographystyle{alpha}
\bibliography{refs}
}

\end{document}

%% file: intro.tex
\section{Introduction}
\label{sec:introduction}

We consider the \kcut problem: given an edge-weighted graph
$G = (V,E,w)$ and an integer $k$, we want to delete a minimum-weight set of edges
so that $G$ has at least $k$ connected components. We let $\lambda_k$ denote the resulting weight of the deleted edges. This
generalizes the global min-cut problem, where the goal is to break the
graph into $k=2$ pieces. 

It was unclear that the problem admitted a
polynomial-time algorithm for fixed $k$, until 
Goldschmidt and Hochbaum gave a deterministic algorithm with
$n^{O(k^2)}$ runtime~\cite{goldschmidt1988polynomial}. The algorithm of Karger~\cite{karger1993global}, based on random edge contractions, can also solve \kcut in $\tO(m n^{2 k-1})$ time; this was later improved to $\tO(n^{2 k - 2})$ runtime by Karger and Stein~\cite{KS96}. There have been a number of improved deterministic algorithms~\cite{GH94,KYN06,Thorup08,chekuri2018lp}: notably, the tree-packing
result of Thorup~\cite{Thorup08} was sped up by Chekuri et
al.~\cite{chekuri2018lp} to $O(mn^{2k-3})$ runtime. Thus, until
recently, randomized and deterministic algorithms with very different
approaches have achieved $n^{(2-o(1))k}$ runtime for the problem. (Here and subsequently, the
$o(1)$ in the exponent indicates a quantity that goes to zero as $k$
increases.)

As for hardness, there is a reduction from $\kkclique$ to $\kcut$. It
is conjectured that solving \kclique requires
$\Omega(n^{(1-o(1))k})$ time when weights are integers in the
range $[1,\Om(n^k)]$, and $\Omega(n^{(\omega/3 - o(1)) k})$ time for
unit weights, where  $\omega$ is the matrix multiplication constant.
Extending these bounds to $\kcut$ suggests that $n^{(1-o(1))k}$ may be a lower bound for general
weighted $k$-cut instances.

There has been recent progress on this problem, showing the following
results:
\begin{enumerate}
\item Gupta, Lee \& Li gave an $n^{(1.98 + o(1))k}$-time algorithm for general
  $\kcut$~\cite{GLL19}. This was based on showing an extremal bound for
  the number of ``small'' $2$-cuts in the graph. A  bounded-depth search is then used to guess the small $2$-cuts within a
 minimum $k$-cut and make progress. This proof-of-concept
  result showed that  $n^{(2-o(1))k}$ was not the right
  bound, but the approach did not seem to extend to exponents considerably below $2k$.
\item For polynomially-bounded edge-weights, Gupta, Lee \& Li gave an
  algorithm with  roughly
  $k^{O(k)} \, n^{(2\omega/3 + o(1))k}$ runtime~\cite{GLL18focs}. For
  unweighted graphs, Li obtained $k^{O(k)} n^{(1+o(1))k}$
  runtime~\cite{Li19focs}. These algorithms are both based on
  finding a spanning tree which crosses a small number of edges of a minimum $k$-cut.
  The former relies on matrix multiplication ideas, and
  the latter on the Kawarabayashi-Thorup graph decomposition \cite{kawarabayashi2018deterministic}, which are both intrinsically tied to graphs with small edge-weights.
\end{enumerate}


In this paper, we 
show that the ``right'' algorithm, the original Contraction Algorithm of Karger~\cite{karger1993global},
achieves the ``right'' bound for general graphs.  We recall the algorithm below; here, $\tau$ (the final desired graph size) is a parameter we will adjust in our specific constructions.
\begin{algorithm}[H]
    \caption{Contraction Algorithm}
    \label{euclid}
    \begin{algorithmic}[1] 
            \While{$|V| > \tau$}
                        \State Choose an edge $e \in E$ at random from $G$, with probability proportional to its weight.
                \State Contract the two vertices in $e$ and remove
                self-loops.
            \EndWhile
            \State Return a $k$-cut of $G$ chosen uniformly at random. 
    \end{algorithmic}
\end{algorithm}

Setting $\tau=k$, as in Karger's original algorithm, would seem most natural,  but we will require a larger value in our analysis.  Our main result is the following.

\begin{restatable}[Main]{theorem}{Main}
\label{thm:main} For any integer $k\geq 2$ and real number $\alpha\ge1$, the Contraction Algorithm outputs each $k$-cut of weight $\alpha
\lambda_k$ with probability at least $n^{-\alpha k} k^{-O(\alpha k^2)}$ for appropriate choice of $\tau=\poly(\alpha,k)$. 
\end{restatable}

Since any minimum-weight $k$-cut (corresponding to $\alpha = 1$) is output with probability
$n^{-k} k^{-O(k^2)}$, this immediately implies the following corollary.

\BC[Number of Minimum $k$-cuts]\label{cor:enum} For any $k \geq 2$, 
the number of minimum-weight $k$-cuts in a graph is at most $n^k k^{O(k^2)}$.
\EC
This improves on the previous best bound of $n^{(1.98 + o(1))k}$~\cite{GLL19}. 
It is almost tight because the cycle on $n$ vertices has $\binom{n}{k}$ minimum $k$-cuts.

Also, while the direct implementation of Algorithm~\ref{euclid} incurs an extra $O(n^2)$ in the runtime,
the Recursive Contraction Algorithm of Karger and Stein~\cite{KS96} can be used to get 
an almost-matching running time to enumerate all minimum $k$-cuts.

\BT[Faster Algorithm to Find a Minimum $k$-cut]\label{cor:runtime}
There is an algorithm to enumerate all minimum $k$-cuts in time $n^{k}(\log n)^{O(k^2)}$ with probability at least $1 - 1/\poly(n)$.
\ET
This improves the runtime 
$n^{(1.98 + o(1))k}$ from~\cite{GLL19}
and even beats the runtime $n^{(1+o(1))k}$ for the unweighted case~\cite{Li19focs}. It is almost optimal under the hypothesis that \kclique requires $n^{(1-o(1))k}$ time.  Achieving an $O(n^{ck})$-time algorithm for unit-weighted graphs for any constant $c < 1$ still remains an open problem. 

See Section~\ref{sec:final} for the formal statements of the above
theorems.

%% file: techniques.tex
\subsection{Our Techniques}
Although we have stated the general $\kcut$ problem for a weighted graph, we will assume throughout that $G = (V,E)$ is an unweighted multigraph with $n$ vertices and $m$ edges. The viewpoint in terms of weighted graphs is equivalent via replicating edges;  note that, in this case,  $m$ may be exponentially large compared to $n$. Our computational and combinatorial bounds will depend on $n$ and not directly on $m$.

In the spirit of \cite{GLL19}, our proof has two main parts:
 (i)~a bound on the extremal number of ``medium'' cuts
 in a graph, and (ii)~a new algorithmic analysis 
 for the Contraction Algorithm. To begin, let us first
state a crude version of our extremal result. Define $\lk:=\la_k/k$, which we think of as the average contribution of the $k$ components of a minimum $k$-cut, and let ``medium'' cuts denote $2$-cuts whose weight is in $[\lk, 2\lk)$.\footnote{In the actual analysis, we use the interval $[\frac{k}{k-1}\lk, 2\lk)$.}
The graph may contain a negligible number of ``small'' $2$-cuts of weight less than $\lk$.
Loosely speaking, the extremal bound says the following:
\begin{quote}
  $(\star)$ For fixed $k$, the graph has at most $O(n)$ many medium cuts. 
\end{quote}

To develop some intuition for this claim, it is instructive to consider the
cycle and clique graphs. These are  two opposite ends of the spectrum in the context
of graph cut. 
In the cycle, we have $\lk=1$, and there are \emph{no} 2-cuts 
with weight less than $2 \lk$, hence $(\star)$ holds. However, the $\bn n2$ minimum 2-cuts
have size \emph{equal to} $2\lk=2$.  In the clique, the minimum $k$-cut chops off $k-1$
singleton vertices, so $\la_k = \binom{k-1}{2} + (k-1) (n-k+1)$, which gives  $\lk \approx \frac{k-1}{k} n$ for $n \gg k$.  There are $n$ minimum 2-cuts, which have weight $n-1 < 2 \lk$ (the singletons), so again $(\star)$ holds. And
again, there are $\bn n2$ 2-cuts of weight approximately $2 \lk$ (the doubletons). 

Therefore, in both the cycle and the clique, the bound $2 \lk$ is almost the best
possible. Moreover, the $O(n)$ bound for the number of medium cuts is also
optimal in the clique.

\subsubsection{Analysis of Contraction Algorithm}
When we begin the Contraction Algorithm in the graph $G$, our extremal bound ensures that there are at most $O_k( n )$ medium cuts of size between $\lk$ and $2 \lk$, plus a negligible number of small 2-cuts of size less than $\lk$. Let us next sketch how these bounds give rise to the improved bound for $k$-cuts.   To provide intuition, let us suppose that in fact there are $n$ medium cuts and no smaller 2-cuts (the precise factors are not important for the overall analysis).  

Each vertex during the Contraction Algorithm corresponds to a 2-cut of the original graph, and we are assuming that $G$ has no small $2$-cuts, so the number of edges in each iteration $i$ of the Contraction Algorithm  is lower-bounded by  $i \lk/2$. Again, to provide intuition, let us suppose there are precisely this many edges.  Then each medium cut gets an edge selected in iteration $i$, and is thereby removed from the graph, with probability at least $\frac{\lk}{i \lk/2} = \frac{2}{i}$.  So after $n/2$ iterations,  the number of surviving medium cuts is close to 
$$
n \prod_{i=n/2}^n \Bigl(1 - \frac{2}{i}\Bigr) \approx n/4.
$$

Thus, in the resulting subgraph with $i = n/2$ vertices,  at most $n/4$ of the vertices (corresponding to the surviving medium cuts) have degree $\lk$. The remainder have degree at least $2 \lk$.   Continuing this process, the graph
 becomes more and more enriched with high-degree vertices. After
 $(1 - \epsilon) n$ iterations (for some small constant $\epsilon$),
 almost all of the medium cuts have been eliminated, and each graph on $i
 \leq \epsilon n$ vertices has close to  $i \lk$ edges.

Now consider an arbitrary minimum $k$-cut $K$. It survives the
first $(1-\epsilon) n$ iterations with constant probability. In each
iteration $i$ of the Contraction Algorithm when the resulting subgraph
has $i \leq \epsilon n$ vertices, $K$ is selected with probability
roughly $ \frac{\lambda_k}{i \lk}
= \frac{k}{i}$. Over the entire run of the Contraction Algorithm, down
to the final graph with $\tau = \poly(k)$ vertices, $K$ survives with
probability roughly
$$
\text{constant} \cdot \prod_{i=\tau}^{\epsilon n} \Bigl(1 - \frac{k}{i}\Bigr)  \approx \Theta(n^{-k}).
$$

To show this formally, we need to track the number of
medium cuts remaining in the residual graphs produced by the
Contraction Algorithm. There are two main obstructions to turning the analysis we have sketched above into a rigorous proof.  First, many of our bounds made unwarranted assumptions
about the parameter sizes; for example, we only know \emph{lower
  bounds} on the edge counts, and we should not assume that these hold
with equality in each iteration. Second, the Contraction Algorithm is
a \emph{stochastic} process; we cannot assume that relevant
quantities (such as the number of medium cuts) equal their expectations.

To overcome these challenges, we adopt a proof strategy of \cite{HarrisSrinivasan}. First, using a number of heuristic worst-case assumptions, and relaxing the discrete stochastic process to a continuous-time system of differential equations, we make a \emph{guess} as to the correct dynamics of the Contraction Algorithm. This gives us a formula for the probability that $K$ is selected, given that the process has reached some iteration $i$ and currently has some given number of residual medium cuts. Next, we use induction to prove that this formula holds in the worst case. For this, we take advantage of the fact that our guessed formula has nice convexity and monotonicity properties.

Let us contrast our proof strategy with the analysis in a preliminary
version of this paper~\cite{GuptaLL20}. In this work, we analyze the
Contraction Algorithm as edges are contracted one at a time. In contrast,
\cite{GuptaLL20} considered an alternate viewpoint
where each edge is \emph{independently} contracted with some given
probability, which is equivalent to executing \emph{many steps}
of the Contraction Algorithm. (The alternate viewpoint is only taken for the purposes of analysis; the actual algorithm remains the same.)  

In some ways, the alternate viewpoint is simpler, since it preserves
many independencies among edges and since a number of relevant
parameters are concentrated. However, a drawback is that it lacks fine
control of precisely how many edges to contract. When the number of
vertices in the graph becomes small, the independent-contractions viewpoint introduces larger errors compared to our one-at-a-time approach.
For example, the preliminary version showed a bound of $n^{k} k^{O(k^2 (\log \log n)^2 )}$ on the number of $k$-cuts; compare
this to the tighter bound of $n^{k} k^{O(k^2)}$ from \Cref{thm:main}.

\subsubsection{Extremal Result}
Recall our target extremal statement $(\star)$: there are $O_k(n)$
many medium cuts in the graph, i.e. $2$-cuts of weight less than $2 \lk$.  To show this, we consider two different cases.

In the first case, suppose the medium cuts all correspond to small vertex sets.   Our key observation is that the $k$-cut structure of the graph  forbids
certain types of sunflowers in the set
family corresponding to the medium cuts; on the other hand, estimates from the Sunflower Lemma would ensure that if there are many medium cuts, then such a sunflower would be forced to exist.

  For, consider a $k$-sunflower of medium cuts
$S_1,S_2,\lds,S_k$, in which the core $C$ is a 2-cut of weight at least $\frac{k}{k-1} \lk$.  (Handling cases where the core is empty or corresponds to a smaller 2-cut are details we defer to the actual proof.)  Suppose we contract $C$ as
well as each petal $P_i = S_i \setminus C$ to single vertices $c$ and
$p_i$ respectively. To provide intuition, let us suppose that there are the same number of edges $r$ between the core and each petal, and let $a \geq \frac{k}{k-1} \lk$ denote the degree of $c$ itself; clearly $r \leq a/k$. See \Cref{fig:1}~right.

Since each set $S_i$ is a medium cut, there are less than $2 \lk$ edges from $\{c, p_i \}$ to $V \setminus \{c, p_i \}$. So $\deg(p_i) < 2 \lk - a +
2 r$ for all $i$ and consequently,  the $k$-cut
$\{p_1, \dots, p_{k-1}, V \setminus \{p_1, \dots, p_{k-1} \} \}$  has
 weight at most $\sum_{i=1}^{k-1} \deg(p_i) < (k-1) \cdot
(2 \lk -a + 2 r) \leq (k-1)(2 \lk - a + 2 a/k)$. Due to our bound on $a$, this is at most $k \lk = \la_k$; this is a contradiction since $\la_k$ is the minimum $k$-cut value.

In the second case, suppose there is a medium cut $S$ where both halves involve many vertices.  Then consider a maximal sequence of medium cuts $S_1, \dots, S_{\ell}$ starting with $S_1 = S$,  such that the Venn diagram of $S_1, \dots, S_{\ell}$ has at least $2 \ell$ regions.  See \Cref{fig:1}~left. From this, we can form two subgraphs where every atom of the Venn diagram of $S_1, \dots, S_{\ell}$ in each half of $S$ gets contracted to a single vertex.  It can be shown that every medium cut of the original graph is preserved in at least one of the two graphs. Also, the fact that both halves of $S$ have many vertices ensures that the contracted graphs are strictly smaller than the original graph.  We get our desired bound by induction on $n$.

\begin{figure}
\centering
\includegraphics[scale=.4]{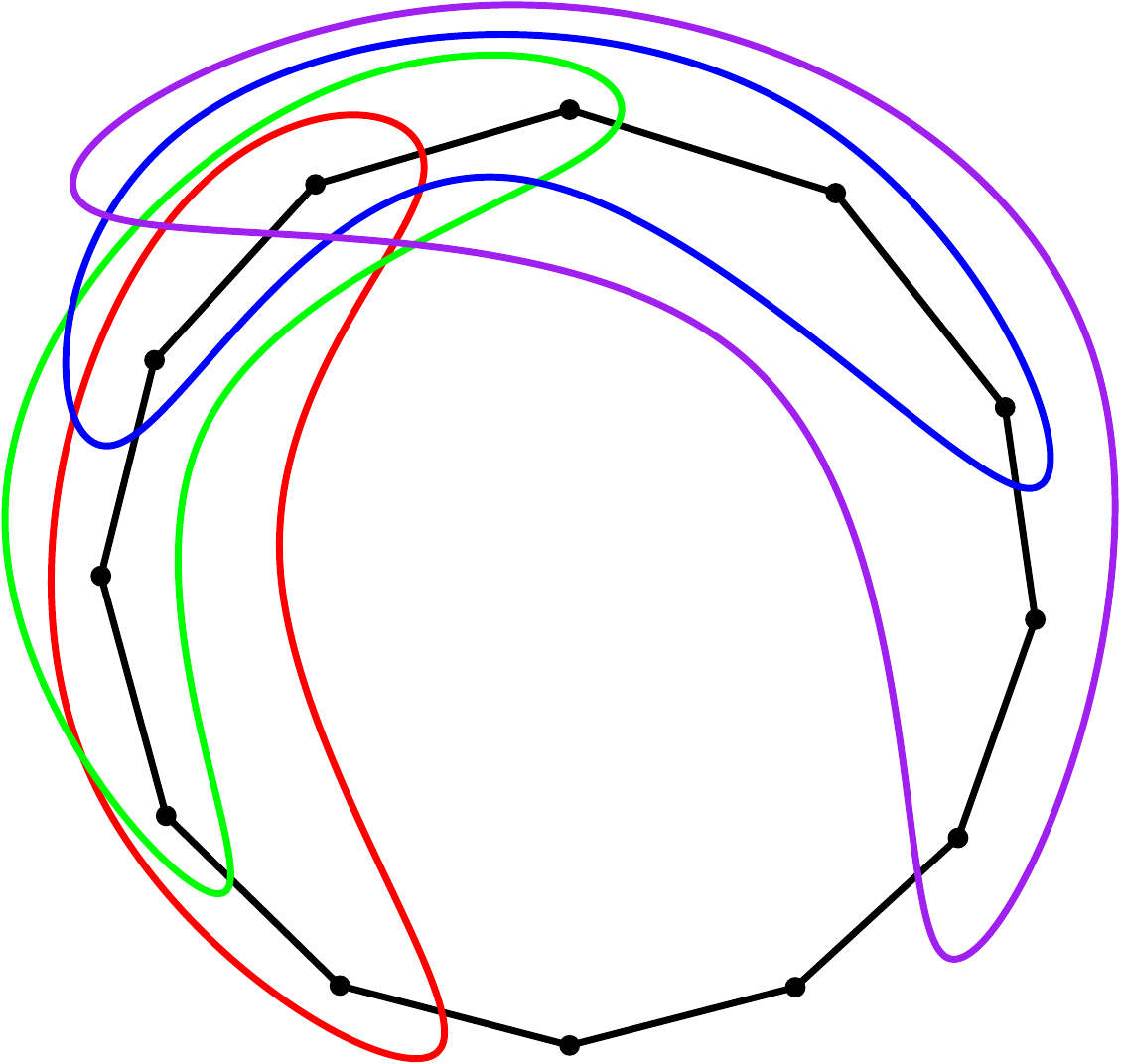}
\qquad
\includegraphics[scale=1]{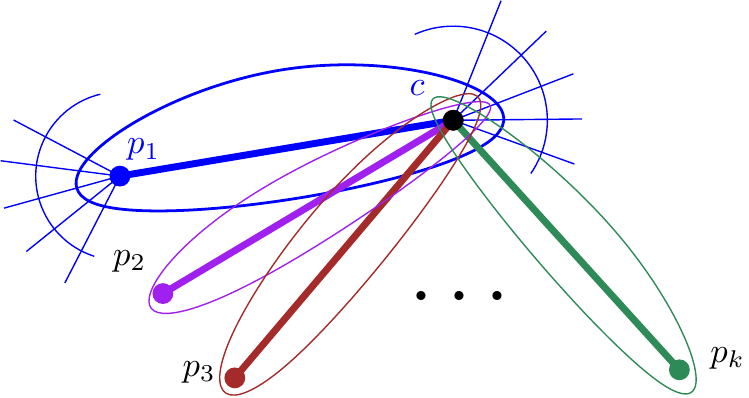}
\caption{\small \emph{Left: To illustrate, suppose $k = 8$ and all $\bn n2$ 2-cuts
  of the cycle have weight less than $2 \lk$. Then, we select $\ell=4$
 such 2-cuts as shown. Their Venn diagram has $2 \ell = 8$ nonempty atoms
  and form an $8$-cut with cost less than 
  $\ell \cd 2 \lk = 8\lk=\la_k$. \quad Right: 
  A $k$-sunflower with core and petals
  consisting of single vertices. Here $c$ has degree $a  \geq\frac{k}{k-1}\lk$ and each  bolded edge has weight $r = a / k = \frac{\lk}{k-1}$. A $k$-cut generated by $k-1$ of the vertices $p_i$ then 
  has weight less than $\lambda_k$.
  }}
\label{fig:1}
\end{figure}



\subsection{Outline}

In \Cref{sec:sunflower}, we discuss the Sunflower Lemma. For our result, we need a slightly strengthened version of this lemma, which involves showing the existence of \emph{multiple sunflowers} and ensuring their cores are nonempty.

In \Cref{sec:extremal-bounds-cuts}, we record some elementary bounds and definitions of cuts and $k$-cuts in the graph. In \Cref{sec:extremal}, we use these for our main extremal bound on the number of medium cuts.

In \Cref{sec:contraction-process}, we provide an overview of the Contraction Algorithm and some simple bounds on the probability that cuts survive it. In \Cref{sec:small-cuts}, we carry out the more involved analysis of how the number of medium cuts evolves during the Contraction Algorithm.

In \Cref{sec:final}, we conclude with our main results on the behavior of the Contraction Algorithm and the Recursive Contraction Algorithm.

%% file: relnote.tex
\section{Sunflower Lemma and Extensions}
\label{sec:sunflower}
In a set system $\m F$ over a universe $U$, an {\em
  $r$-sunflower} is a collection of $r$ sets $F_1, \dots, F_r \in \m
F$ which all share the same pairwise intersection. That is, there is a {\em core} $C\subseteq U$ such that $F_i \cap F_j = C$ for all $i, j$, and hence $\bigcap_i F_i = C$.
Let $\sun(d, r)$ be the smallest number such that any set system with more than $\sun(d, r)$ sets of cardinality at most $d$ must have an $r$-sunflower.
The classical bound of Erd\H{o}s and Rado~\cite{ER60} shows that $\sun(d, r) \leq d!(r - 1)^d$.
A recent breakthrough by Alweiss et al.~\cite{alweiss2019improved} shows
that
\begin{gather}
  \sun(d, r) \leq (\log d)^d (r \cdot \log \log d)^{O(d)}. \label{eq:sunf}
\end{gather}
While we use this improved bound, it only changes 
lower-order terms: the older Erd\H{o}s-Rado bound would give the same
asymptotics for our applications.


For our applications for cuts, we want multiple sunflowers with distinct nonempty cores. (The cores may intersect, even though they are distinct.)  The bound must then depend on the universe size $N$, since the system consisting of $N$ singleton sets has no sunflowers with nonempty core. The following results show that we can guarantee a nonempty core by multiplying the bound by $N$.

\begin{proposition}
\label{sunflower}
Let $\m F$ be a family of nonempty sets over a universe of $N$ elements, where every set has size at most $d$. If $|\m F|> \sun(d, r) \cdot N$, then $\m F$ contains an $r$-sunflower with nonempty core.
\end{proposition}
\BP
For each element $v$ of the universe, consider the set system $\m F_v:=\{F\in\m F:F\ni v\}$. Since every set in $\m F$ is included in some $\m F_v$, there must be some element $v$ with $ |\m F_v| \geq |\m F| / N > \sun(d,r)$. Thus, there is an $r$-sunflower in $\m F_v$ and hence $\m F$.  The core is nonempty since it contains $v$.
\EP

\BL\leml{sunflower-nonempty}
Let $\m F$ be a family of nonempty sets over a universe of $N$ elements, where every set has size at most $d$. If $|\m F|> \sun(d, r) \cdot  sN$, then $\m F$ contains $s$ many $r$-sunflowers, each with a distinct, nonempty core.
\EL
\BP
We show this by induction on $s$. The base case $s=0$ is vacuous. For the induction step with $s \geq 1$, consider a maximal nonempty set
$C$ such that $\m F$ contains an $r$-sunflower with core
$C$; this exists by \Cref{sunflower} since $|\mathcal F| > \sun(d,r) \cdot sN \geq \sun(d,r) N$.

We claim that the set system $\m F_C:=\{F\in\m F:F\supseteq C\}$
has size at most $\sun(d, r) \cdot N$. For, if not, then
applying~\Cref{sunflower} to the set system $\{F \sm C:F\in\m
F_C \}$ (which has the same cardinality as $\m F_C$) would give
an $r$-sunflower $S_1,\lds,S_r$ with nonempty core
$C'$. The sets $S_1\cup C,\lds,S_r\cup C$ in $\m F$ then form an
$r$-sunflower with core $C\cup C'$, contradicting maximality of $C$.

Now consider the set system $\mathcal F' = \mathcal F \setminus \mathcal F_C$. It has size $|\mathcal F| - | \mathcal F_C| > \sun(d,r) \cdot sN - \sun(d,r) \cdot N = \sun(d,r) \cdot (s-1) N$.  By the induction hypothesis, it has $s-1$ many $r$-sunflowers with distinct nonempty cores. These cores are all distinct from $C$, since no sets containing $C$ remain in $\mathcal F'$. Combining them with the $r$-sunflower of core $C$ gives  $s$ many $r$-sunflowers with distinct, nonempty cores.
\EP

\section{Simple Bounds and Definitions for Cuts}
\label{sec:extremal-bounds-cuts}

We assume throughout we have a fixed value $k \geq 3$.  A $k$-cut $K$ is a partition of $V$ into $k$ nonempty sets,
and we let $\partial K$ denote the set of edges crossing different parts of $K$.  The \emph{weight} of $K$ is the cardinality of the edge set $\partial K$.  We let $\lambda_k$ be the minimum weight of any $k$-cut, and $\lk := \lambda_k / k$.  

A $2$-cut $\{C, V \setminus C \}$ will often simply be called a
\emph{cut}, and we often denote it merely by $C$.   The \emph{shore} of the cut is whichever of the sets $C$ or $V \setminus C$ is smaller. (If they are the same size, choose one arbitrarily), and the \emph{shoresize}  is the cardinality of the shore.

For vertex sets $A, B$ we let $E(A,B)$ denote the set of edges crossing from $A$ to $B$. We also write $\partial S = E(S, V \setminus S)$ for a set $S \subseteq V$.

We define a \emph{small cut} to be a cut $C$ with $$
|\partial C| < \frac{k}{k-1} \lk,
$$
 and we define a \emph{medium cut} to be a
cut $C$ such that
$$
  \frac{k}{k-1} \lk \leq |\partial C| < 2 \lk. 
$$

Given vertex sets $F_1, \dots, F_{t}$, we denote their {\em Venn diagram} by
$\Venn(F_1, \lds, F_t)$.  An {\em atom} denotes a nonempty region of the
diagram, i.e., a nonempty set that can be expressed as
$G_1 \cap \lds \cap G_t$, where each set $G_i$ is either
 $F_i$, or its complement $V \setminus F_i$. 
 See \Cref{fig:venn}.
 
\begin{figure}[H]\centering
\includegraphics[scale=0.6]{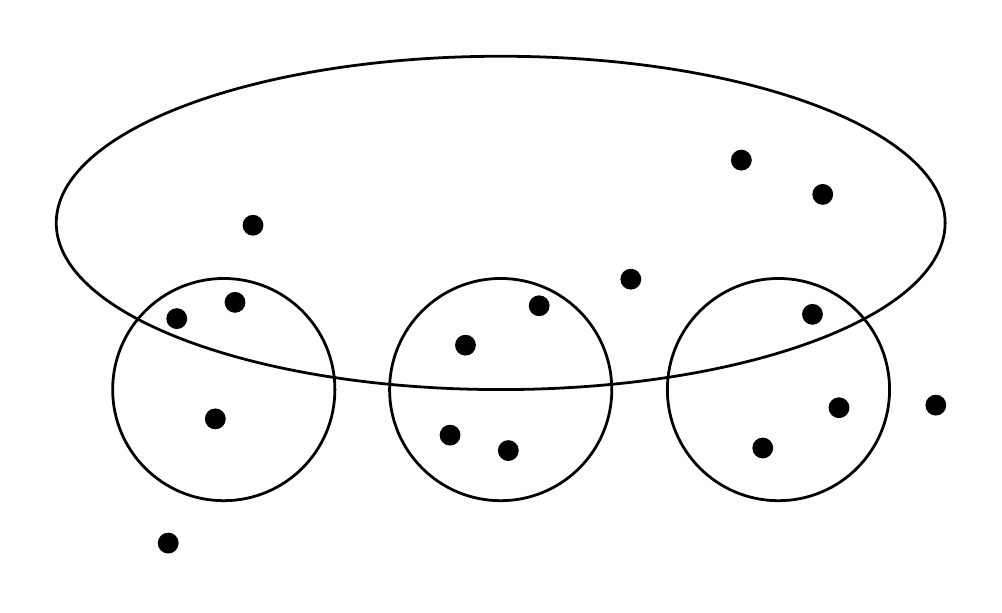}
\caption{The Venn diagram above has eight atoms.}
\label{fig:venn}
\end{figure}

We say that $F_1, \dots, F_t$ \emph{generate} the $\ell$-cut $K = \{A_1, \dots, A_{\ell} \}$ where $A_1, \dots, A_{\ell}$ are the atoms of $\Venn(F_1, \dots, F_t)$. Observe that the weight of $K$ is at most $|\partial F_1 | + \dots + |\partial F_t|$.

We begin with a few straightforward bounds.

\begin{proposition}
 \label{m-bound0}
 If $n \geq k$, then $m \geq \frac{n k}{2(k-1)} \lk$.
\end{proposition}
\begin{proof}
Sort the vertices in ascending order of degree, so $\deg(v_1) \leq \deg(v_2) \leq \dots \leq \deg(v_n)$.  
The $k$-cut generated by the singleton sets $\{v_1 \}, \dots, \{v_{k-1} \}$ has weight at most $\deg(v_1) + \dots + \deg(v_{k-1})$; since $\lambda_k$ is the minimum $k$-cut, we thus have $\deg(v_1) + \dots + \deg(v_{k-1}) \geq \lambda_k$. Also, because of the sorted vertex ordering,  we have $\deg(v_i) \geq \deg(v_{k-1}) \geq \frac{\deg(v_1) + \dots + \deg(v_{k-1})}{k-1} \geq \frac{\lambda_k}{k-1}$ for all $i \geq k$.

Summing vertex degrees, the total number of edges $m$ is given by
\begin{align*}
2m &= \bigl( \deg(v_1) + \dots + \deg(v_{k-1}) \bigr) + \bigl( \deg(v_k) + \dots + \deg(v_n) \bigr) \\
&\geq \lambda_k + (n-k+1) \cdot \lambda_k / (k-1) = n \lambda_k / (k-1) = n k \lk / (k-1). \qedhere
\end{align*}
\end{proof}

\begin{lemma}\label{lemma:small}
There are fewer than $2^{k-2}$ small cuts.
\end{lemma}
\begin{proof}
Suppose not; in this case, we will construct a $k$-cut of weight less than $\la_k$, which contradicts the definition of $\la_k$.

For $i = 1, \dots, k-1$, let us choose an arbitrary small cut $S_i$ such that $|\Venn(S_1,\ldots,S_i)| \geq i + 1$. We claim that we can always find such an $S_i$.  For, if $|\Venn(S_1,\ldots,S_{i-1})| \geq i + 1$, then $S_i$ can be chosen arbitrarily. Otherwise, suppose that $\Venn(S_1,\ldots,S_{i-1})$ has precisely $i$ atoms $A_1, \dots, A_i$. The only small cut $T$ such that $|\Venn(S_1,\ldots,S_{i-1},T)| = i = |\Venn(S_1,\ldots,S_{i-1})|$ would have the form $T = \bigcup_{j \in I}A_j$ for some subset $I\s \{1, \dots i \}$. There are at most $2^{i-1} - 1$ such cuts (keeping in mind that $I$ and its complement determine the same cut). Since by assumption there are at least $2^{k-2}$ small cuts, there exists a small cut $S_i$ with $\Venn(S_1,\ldots,S_{i-1}, S_i) > i$ as desired.

At the end, we have $|\Venn(S_1,\ldots,S_{k-1})| \geq k$. So the small cuts $S_1, \dots, S_{k-1}$ generate a $t$-cut for $t \geq k$ whose weight is less than $(k-1) \cdot \frac{k}{k-1} \lk  =\la_k$. This is our desired contradiction.
\end{proof}

\begin{proposition}
\label{ell-bound}
Let $T_1, \dots, T_{r}$ be medium cuts where $r = \lceil k/2 \rceil$. Then either $|\Venn (T_1, \dots, T_{r-1})| < 2(r-1)$ or $|\Venn (T_1, \dots, T_{r})| < 2 r$ (or both).
\end{proposition}
\begin{proof}
Let us first consider the case where $k$ is even and $r = k/2$. Suppose for contradiction that $|\Venn (T_1, \dots, T_{r})| = t \geq k$. Then $T_1, \dots, T_{r}$ generate a $t$-cut $K$. Since $T_1, \dots, T_{r}$ are medium cuts, the weight of $K$ is less than $r \cdot 2 \lk = \lambda_k$; this contradicts that $\lambda_k$ is the minimum $k$-cut value.

Next consider the case where $k$ is odd and $r = (k+1)/2$. Suppose for contradiction that $|\Venn (T_1, \dots, T_{r})| = t \geq k+1$  and $|\Venn (T_1, \dots, T_{r-1})| = t' \geq k-1$. The sets $T_1, \dots, T_{r-1}$ generate a $t'$-cut $K'$; since $T_1, \dots, T_{r-1}$ are medium cuts, the weight of $K'$ is less than $(r-1) \cdot 2 \lk = \frac{k-1}{k} \lambda_k$. If $t' \geq k$, this contradicts that $\lambda_k$ is the 
minimum $k$-cut value.  So it must be that $t' = k-1$ exactly. 

Let $A_1, \dots, A_j$ be the atoms of $\Venn (T_1, \dots, T_{r-1})$ cut by $T_{r}$; since $t \geq k+1$ and $t' = k-1$ we must have $j \geq 2$. The edge sets $E(T_{r}, A_i \setminus T_r)$ are all disjoint and $T_{r}$ is a medium cut, so at least one atom $A_i$ must satisfy $|E(T_{r}, A_i \setminus T_r) | \leq | \partial T_{r} | / j \leq | \partial T_{r} | / 2 \leq \lk$. The sets $T_1, \dots, T_{r-1}, A_i$ then generate a $k$-cut $K''$ of weight less than  $(r-1) \cdot 2 \lk + \lk = \lambda_k$, contradicting that $\lambda_k$ is the minimum $k$-cut value.
\end{proof}

\section{Bounding the Number of Medium Cuts}
\label{sec:extremal}
 We now analyze the combinatorial structure of the medium cuts  to show the following key bound:
\begin{theorem}
\label{main-extremal-theorem}
There are $k^{O(k)}n$ many medium cuts.
\end{theorem}

We prove this in two stages. First, using the Sunflower Lemma, we show it for the
special case when all the medium cuts of $G$ have shoresize at most
$k$. We then extend to the general case by an induction on the graph size. 

\begin{lemma}
\label{lsf1}
Suppose the medium cuts all have shoresize at most $k$. Then there are at most $k^{O(k)} n$ medium cuts.
\end{lemma}
\begin{proof}
Let $\mathcal F$ be the set family consisting of the shores of the medium cuts.  We claim that $\mathcal F$ cannot have $2^k$ many $k$-sunflowers with distinct nonempty cores.  For, suppose for contradiction that it does so.  Then, by \Cref{lemma:small}, at least one of the
sunflowers has a nonempty core $C$ with $|\partial C|\ge\frac
k{k-1}\bar\lambda_k$.   Let the sets in this sunflower be 
$S_1,\ldots,S_{k}\in\mathcal F$ where $S_i \cap S_j = C$ for $i \neq j$.

Let $P_i := S_i \setminus C$ be the \emph{petal} for each $S_i$, and let $L_i = E(P_i, C)$ denote the set of edges between $P_i$ and $C$. By inclusion-exclusion, we have $|\partial P_i| =  |\partial S_i |-|\partial C| +2 | L_i |$. Since each $S_i$ is a medium cut,  it satisfies $| \partial S_i | < 2 \lk$, so $$
| \partial P_i | < 2 \lk  -|\partial C| +2 | L_i |.
$$

Suppose the petals are sorted in ascending order of $|L_i|$, so that $|L_i| \leq |L_{i+1}|$ for $i = 1, \dots, k-1$. Consider the $k$-cut $K$ generated by the disjoint sets $P_1, \dots, P_{k-1}$. We can bound its weight $|\partial K|$ by:
$$
|\partial K| \leq \sum_{i=1}^{k-1} |\partial P_i| < \sum_{i=1}^{k-1} (2 \lk  -|\partial C| +2 | L_i |) = 2 (k-1) \lk - (k-1) |\partial C| + 2 \sum_{i=1}^{k-1} |L_i|.
$$

Because the sets $L_i$ are pairwise disjoint subsets of $\partial C$ in sorted order of size, we have 
$$
\sum_{i=1}^{k-1} |L_i| \leq \frac{k-1}{k} \sum_{i=1}^{k}  |L_i| \leq \frac{k-1}{k} |\partial C|,
$$
and so 
$$
|\partial K| < 2 (k-1) \lk - (k-1) |\partial C| + 2 \cdot \tfrac{k-1}{k} |\partial C| = 2 (k-1) \lk- \tfrac{(k-1)(k-2)}{k} 
|\partial C|.
$$

Finally, using the bound $|\partial C| \geq \frac{k}{k-1} \lk$, we get
$$
|\partial K| < 2 (k-1) \lk - \tfrac{(k-1)(k-2)}{k} \cdot \tfrac{k}{k-1} \lk = k \lk = \lambda_k.
$$

This contradicts the definition of $\lambda_k$ as the minimum $k$-cut. Thus $\m F$ cannot have $2^k$ many $k$-sunflowers with distinct, nonempty cores. By our hypothesis, the sets in $\m F$ have size at most $k$.  Thus, by \lem{sunflower-nonempty} (with parameters $d = r = k$ and $N = n$ and $s = 2^k$) and Eq.~(\ref{eq:sunf}), this means 
\[
| \m F | \leq \sun(d, r) \cdot s N 
\leq (\log d)^d (r \cdot \log \log d)^{O(d)} \cdot 2^k n
\leq k^{O(k)} n.
\qedhere
\]
\end{proof}

We will next remove the restriction on the shoresize, completing the proof. 
\begin{proof}[Proof of \Cref{main-extremal-theorem}]

We will show by induction on $n$ that for $n > k$ there are at most $c_k (n - k)$ medium cuts in any graph $G$, for some constant $c_k = k^{O(k)}$.

If every medium cut has shoresize at most $k$, then we have already shown this in \Cref{lsf1} for appropriate choice of $c_k$. (This covers the base case of the induction $n = k+1$.) We thus consider a medium cut $S$ with shoresize larger than $k$, i.e. $k < |S| < n-k$.  

Starting with $S_1 = S$, let us form a maximal sequence of medium cuts $S_1, S_2, \dots, S_\ell$  with the property that $|\Venn (S_1, \dots, S_t)| \geq 2t$ for all $t = 1, \dots, \ell$; here $\ell \geq 1$ since $|\Venn (S)| = 2$. Let the atoms of $\Venn (S_1, \dots, S_{\ell})$ inside $S$ (respectively, outside $S$) be $A_1, \dots, A_i$ and $B_1, \dots, B_j$. So $A_1 \cup \dots \cup A_i = S$ and $B_1 \cup \dots \cup B_j = V \setminus S$.    

Now form a graph $H_1$ by contracting each of the atoms $A_1, \dots, A_i$ and likewise form a graph $H_2$ by contracting each of the atoms $B_1, \dots, B_j$.  Since $A_1, \dots, A_i$ partition $S$ and $B_1, \dots, B_j$ partition $V \setminus S$, these graphs have $n_1 = (n- |S|) + i$ and $n_2 = |S| + j$ vertices respectively.  See \Cref{fig:ex} for an example. 

\begin{figure}[h]\centering
\includegraphics[scale=.5]{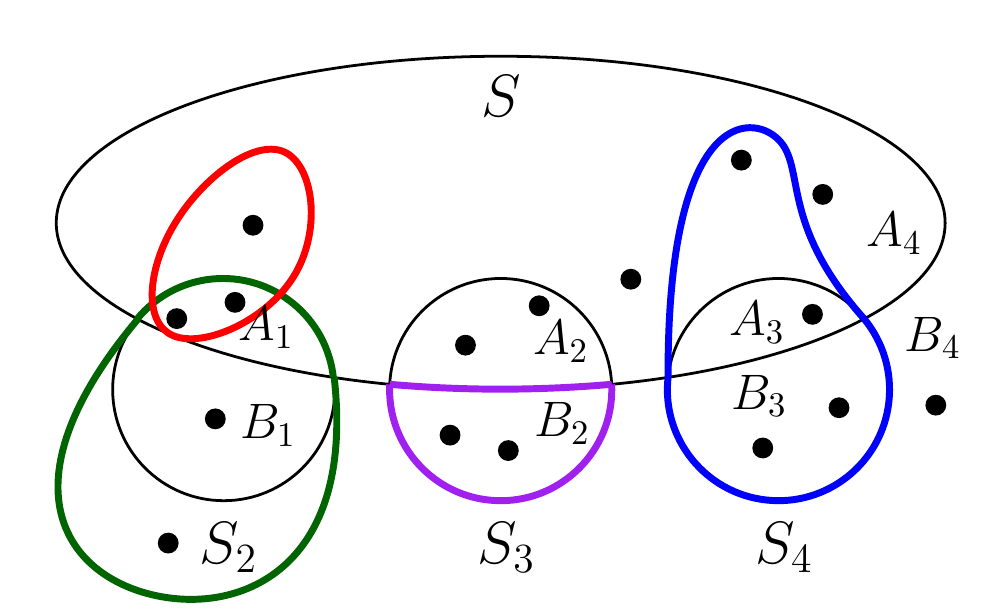}
\includegraphics[scale=.5]{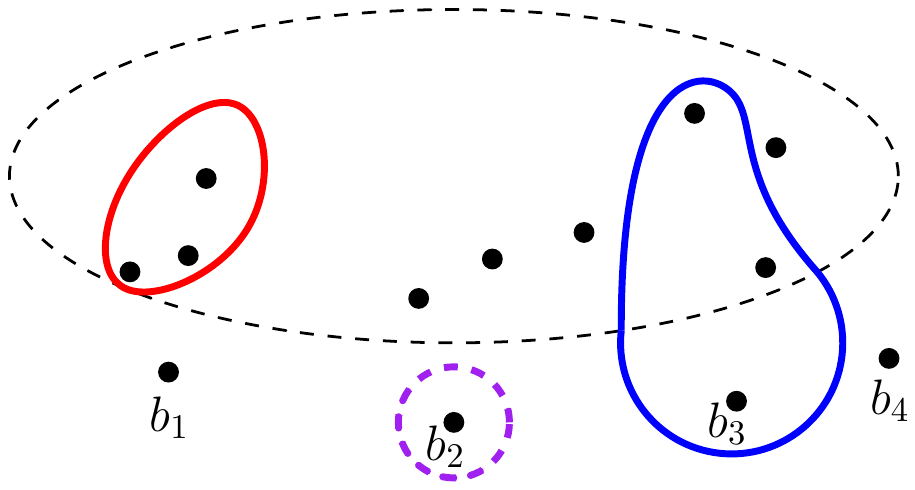}
\includegraphics[scale=.5]{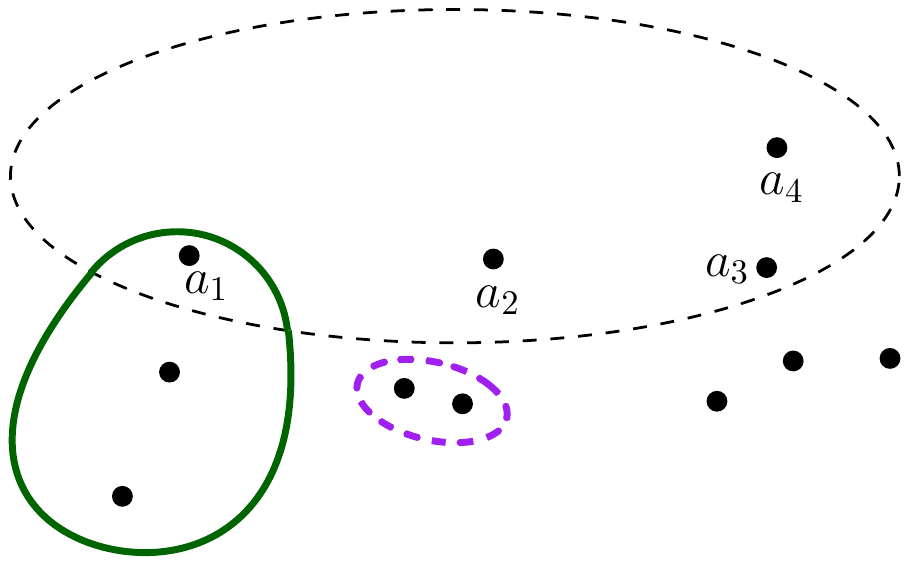}
\caption{Construction of graphs $H_1$ (right) and $H_2$ (middle) given medium cuts $S_1,S_2,S_3, S_4$ (left). 
Each colored set represents a medium cut surviving in either $H_1$ or $H_2$. The red and blue cuts survive in $H_2$, and the green cut survives in $H_1$. The purple cut survives in both $H_1$ and $H_2$.}
\label{fig:ex}
\end{figure}

We claim that $i + j < 2 ( \ell + 1)$ and $\ell < k/2$. For, if $i+j \geq 2(\ell+1)$, then consider choosing $S_{\ell+1}$ to be an arbitrary medium cut;  we would have $|\Venn (S_1, \dots, S_{\ell+1})| \geq i + j \geq 2 (\ell + 1)$, contradicting maximality of $\ell$. Likewise, if $\ell \geq k/2$, then we would have $|\Venn(S_1, \dots, S_{r-1})| \geq 2(r-1)$ and $|\Venn(S_1, \dots, S_r)| \geq 2 r$ where $r = \lceil k/2 \rceil$; this would contradict \Cref{ell-bound}.

From these two bounds, we conclude that $i+j \leq 2 \ell + 1 \leq k$. Since  $k < |S| < n - k$, both $n_1$ and $n_2$ are strictly larger than  $k$ and strictly smaller than $n$. Hence, from the induction hypothesis, the number of medium cuts in $H_1$ and $H_2$ is at most $c_k (n_1 - k)$ and $c_k (n_2 - k)$ respectively.

We now claim that every medium cut of the original graph $G$ survives in either $H_1$ or $H_2$ (or both). For, suppose there is some medium cut $T$ where an edge $e' \in \partial T$ lies inside an atom $A_{i'}$ and an edge $e'' \in \partial T$ lies inside an atom $B_{j'}$. Then the atoms $A_{i'}$ and $B_{j'}$ would both split into two new atoms  in $\Venn (S_1, \dots, S_{\ell}, T)$, giving $|\Venn (S_1, \dots, S_{\ell}, T)| \geq i+j + 2$. This contradicts maximality of $\ell$. 

Consequently, the number of medium cuts in $G$ is at most 
$$
c_k  (n_1  - k) + c_k (n_2 - k) = c_k (  (n - |S| + i) + (|S| + j) - 2 k) = c_k (n + i + j - 2 k).
$$

Now, $i + j \leq k$ so this is at most $c_k (n - k)$, completing the induction.
\end{proof}

\section{The Contraction Process}
\label{sec:contraction-process}
Our next goal will be to lower-bound the probability that a given $k$-cut $K$ is preserved during the Contraction Algorithm. More generally, for an edge set $J \subseteq E(G)$, we say that $J$ \emph{survives} the Contraction Algorithm if no edge of $J$ ever gets selected during any iteration.  Following \cite{HarrisSrinivasan}, we define the \emph{Contraction
  Process up to stage $i$} for $J$ as follows.
\begin{quote}
  Starting with the graph $G_n = G$, in stage $j$ we select an edge
  $e_j$ from the resulting (random) subgraph $G_j$ uniformly at random
  \emph{excluding the edges in $J$ itself}, and contract
  $e_j$ to get the graph $G_{j-1}$. We stop when we reach $G_i$. 
\end{quote}

 It is possible, and allowed, for some edges of $J$ to become self-loops and be removed from the graph. When considering a subgraph $G_j$ during the Contraction Process for $J$,  bear in mind that we may have $J \not \subseteq E(G_j)$.
 
 For the Contraction Process for $J$,  we define the key statistic
\begin{gather}
  R_i = \sum_{j=i+1}^n \frac{\lk}{|E(G_j)|}. \label{eq:Ri}
\end{gather}
Here,  $R_i$ serves as a linearized approximation to the
probability of avoiding $J$ in the Contraction
Algorithm. Specifically, we show the following result which is a slight reformulation of \cite{HarrisSrinivasan}:
\begin{proposition}
  \label{hf1}
Let $J$ be an edge set and let $\alpha = |J|/\lambda_k$. Suppose we run the Contraction Algorithm up to stage $i \geq \max \{4 \alpha k, k \}$. The probability that $J$ survives is at least $e^{-\alpha k \bE[R_i] - \alpha k}$,   where the expectation is taken over the Contraction Process for $J$ up to stage $i$. 
\end{proposition}
\begin{proof}
For $i \leq j \leq n$ let us define $$
 x_{j} = \frac{|J|}{|E(G_j)|} = \frac{\alpha k \lk}{|E(G_j)|}
 $$
 where $G_j$ is the subgraph obtained at stage $j$ of the Contraction Process for $J$ starting at $G$.  We also define the corresponding random variable
$$
L_G = \prod_{j=i+1}^{n} (1 - x_{j}).
$$

Note that, by the property of iterated expectations, we calculate the expected value of $L_G$ as:
\begin{align*}
\bE[L_G] &= \frac{1}{|E(G) \setminus J|} \sum_{e \in E(G) \setminus J} \bE[L_G \mid e_n = e] = \frac{1}{|E(G) \setminus J|} \sum_{e \in E(G) \setminus J} \bE[ (1-x_{n}) L_{G/e}] \\
&= (1 - x_n)  \sum_{e \in E(G) \setminus J}  \frac{\bE[L_{G / e}]}{|E(G) \setminus J|},
\end{align*}
where $G/e$ denotes the graph obtained by contracting edge $e$ in $G$; note that $G_{n-1} = G/e_n$ during the Contraction Process.

We first claim that if we run the Contraction Algorithm on $G$, then $J$ survives to stage $i$ with probability at least $\bE[L_G]$. We show this by induction on $n$. The case $n = i$ holds vacuously since then $L_G = 1$ with probability one and $J$ survives with probability one.  

For the induction step, let $n > i$.  The Contraction Algorithm chooses edge $e_n$ uniformly at random from $E(G)$ and then continues on $G / e_n$. Edge set $J$ survives to stage $i$ if and only if the following events occur: (i) $e_n \in E(G) \setminus J$ and (ii) conditional on fixed choice of $e_n = e$,  the edge set $J$ survives the Contraction Algorithm in $G / e$ to stage $i$.  By the induction hypothesis, the latter event has probability at least $\bE[L_{G/e}]$, and so:
\begin{align*}
&\Pr(\text{$J$ survives starting from $G$})  \geq \frac{1}{|E(G)|} \sum_{e \in E(G) \setminus J} \bE[L_{G / e}] = \frac{|E(G) \setminus J|}{|E(G)|} \sum_{e \in E(G) \setminus J} \frac{ \bE[L_{G / e}] }{ |E(G) \setminus J|}.
\end{align*}
Note now that $$
\frac{|E(G) \setminus J|}{|E(G)|} \geq \frac{|E(G)| - |J|}{|E(G)|}  = 1 - x_{n},
$$
so this is at least $$
(1 - x_{n}) \sum_{e \in E(G) \setminus J}  \frac{\bE[L_{G / e}]}{|E(G) \setminus J|} = \bE[L_G],
$$
 which concludes the induction.

So $J$ survives the Contraction Algorithm with probability at least $\bE[L_G]$.  We need to bound $L_G$. Consider some stage $j \geq i$ of the Contraction Process for $J$. By \Cref{m-bound0}, since $i \geq k$, we have $|E(G_j)| \geq \frac{j k}{2(k-1)} \lk$ so $x_{j} \leq 2 \alpha (k-1)/j.$  Since $j \geq i \geq 4 \alpha k$ this implies $x_{j} \leq 1/2$.   We use the elementary identity $1-x \geq e^{-x - x^2}$ for $x \in [0,1/2]$ to get:
 \begin{align*}
  L_G &= \prod_{j=i+1}^n (1 - x_{j}) \geq \prod_{j=i+1}^n e^{-x_{j} - x_{j}^2} \geq \prod_{j=i+1}^n e^{-\alpha k \lk / |E(G_j)|  - (2 \alpha (k-1)/j )^2} \\
  &= e^{-\alpha k R_i -4 \alpha^2 (k-1)^2 \sum_{j=i+1}^n 1/j^2} \geq e^{-\alpha k R_i -4 \alpha^2 (k-1)^2/i}.
\end{align*}

Since $i \geq 4 \alpha k$, we thus have $L_G \geq e^{-\alpha k R_i-\alpha k}$. Taking expectations and using Jensen's inequality, we have  $\bE \left[ L_G \right] \geq \bE[e^{-\alpha k R_i - \alpha k}] \geq e^{-\alpha k - \alpha k \bE[R_i]}.$
\end{proof}

Using this, we can recover Karger and Stein's original success probability of $n^{-2\alpha(k-1)}$. Although it is much weaker than the bound of $n^{-\alpha k}$ we want, this is useful for a few edge cases in the analysis.

\begin{corollary}
\label{simpl-cor1}
For any parameter $\alpha \geq 1$ and any $k$-cut $K$ with $|\partial K| \leq \alpha \lambda_k$, the Contraction Algorithm with parameter $\tau = \lceil 4 \alpha k \rceil$ selects $K$ with probability at least $n^{-2 \alpha(k-1)} k^{-O(\alpha k)}$.
\end{corollary}
\begin{proof}
Consider the Contraction Process for edge set $J = \partial K$. In each stage $j$, \Cref{m-bound0} shows that graph $G_j$ has at least $\frac{j k}{2(k-1)} \lk$ edges. Hence, with probability one, there holds
$$
R_{\tau} \leq \sum_{j=\tau+1}^n \frac{2(k-1)}{k j}  \leq \frac{2(k-1)}{k} \log (n/\tau).
$$

By \Cref{hf1}, the probability $K$ survives to stage $\tau$ is at least $e^{-\alpha k \bE[R_\tau] - \alpha k}$, which is at least $n^{-2 \alpha (k-1)}$ with our bound on $R_{\tau}$.  Next, suppose that $K$ does survive to stage $\tau$ (this includes the case where $n \leq \tau$). The resulting graph has at most $\tau$ vertices and hence at most $k^\tau$ different $k$-cuts. Thus, $K$ is selected from this graph with probability at least $k^{-\tau} \geq k^{-O(\alpha k)}$. Overall $K$ is selected with probability at least $n^{-2 \alpha (k-1)} k^{-O(\alpha k)}$.
\end{proof}

\section{Analyzing the Dynamics of the Contraction Process}
\label{sec:small-cuts}
Our goal now is to analyze the Contraction Process for a given edge set $J$. Let $\alpha = |J|/\lambda_k$.  We fix some parameter $\eps \in [0,1/k)$, and define a \emph{good cut} to be a medium cut $C$ with  
$$
|\partial C \setminus J| \geq (1-\eps) \frac{k}{k-1} \lk;
$$
 the role of $\eps$ will be explained later. We also define two related parameters 
 $$
 \delta := \frac{1 - \eps k}{k-1},  \qquad \text{and} \qquad \beta := k + 2 \alpha k / \eps.
$$

Note that $\delta > 0$ and $1 + \delta = (1 - \eps) k / (k-1)$. We begin with  a  lower bound on edge count in a single iteration of the  Contraction Algorithm.

\begin{proposition}
  \label{m-bound}
  Let $s$ be the number of good cuts in $G$. If $n \geq \beta$, then
  $$
m \geq   s \cdot \tfrac{k}{k-1} \lk/2 + (n - s - \beta) \cdot \lk.
  $$
\end{proposition}
\begin{proof}
  Each vertex $v$ of $G$ corresponds to a cut $C_v$. At most $k-2$ of these vertex cuts may be small cuts. For, if there are $k-1$ such vertices $v_1, v_2, \dots, v_{k-1}$, then the $k$-cut generated by the $ \{v_1 \}, \dots, \{v_{k-1} \}$  would have weight below $(k-1) \cdot \frac{k}{k-1} \lk  = k \lk = \lambda_k$, a contradiction. 
  
Let $U$ denote the set of vertices $v$ for which cut $C_v$ is medium but not good. For $v \in U$, we have $|\partial C_v| \geq \frac{k}{k-1} \lk$ and hence $|\partial C_v \cap J| \geq \eps \lk$. Each edge appears in at most two vertex cuts, so
  $$
| J | \geq \tfrac{1}{2} \sum_{v \in U}  | \partial C_v \cap J | \geq |U| \eps \lk / 2;
  $$
  since $|J| = \alpha k \lk$, this implies $|U| \leq 2 \alpha k / \eps$. 

Summarizing, at most $k-2$ vertices correspond to small
  cuts, and at most $s + 2 \alpha k / \eps$ vertices correspond to medium cuts. The remaining vertices (at least $n - s - 2 \alpha k/\eps - k + 2$ of them) correspond to large cuts so their degree is at least $2 \lk$. We thus have
  \[
  2 m \geq (s + 2 \alpha k /\eps) \cdot \tfrac{k}{k-1} \lk + (n - s - 2 \alpha k /\eps - k + 2) \cdot 2 \lk \geq s \cdot  \tfrac{k}{k-1} \lk + (n - s - \beta) \cdot 2 \lk. \qedhere
  \]
\end{proof}

For our purposes, we can combine \Cref{m-bound0} and \Cref{m-bound} to get the following (somewhat crude) estimate:
\begin{corollary}
\label{m-bound2}
If $G$ has $s$ good cuts, then $m \geq  (n - \beta) \lk - \min\{s, (n-\beta) \} \lk/2$.
  \end{corollary}

We are now ready derive the key bound on the random variable $R_i$ as defined in Eq.~(\ref{eq:Ri}).  For $p \geq j$ and $s \geq 0$, define the function
\begin{gather}
  f(j,s,p) = \log(p/j) + \frac{\log \Bigl( 1 + (s/p) (1 + 1/\delta) (1
    - (j/p)^{\delta}) \Bigr)}{1 + \delta}. \label{eq:fun-f}
\end{gather}

We will prove a bound on $\bE[R_i]$ in terms of the  function $f$ by induction. The derivation of the  function $f$ is itself rather opaque; we describe the (non-rigorous) analysis which leads to it in \Cref{sec:heuristic-bound}.  We first observe a few analytical properties of function $f$. 
\begin{proposition}
  \label{gb1}
  For $p \geq j$ and $s \geq 0$, we have the following:
  \begin{enumerate}
  \item Function $f(j,s,p)$ is a well-defined, nonnegative, nondecreasing, concave-down function of  $s$.
    \item  The function $y \mapsto y + f(j, s e^{-(1+\delta) y}, p)$ is an increasing function of $y$.
  \end{enumerate}
\end{proposition}
\begin{proof}
  \begin{enumerate}
  \item
    The argument of the logarithm in function $f$ is an affine function of $s$, with constant term $1$ and coefficient $ \frac{1}{p} (1 + 1/\delta) (1 - (j/p)^{\delta}) \geq 0$.        
    
  \item  The derivative as a function of $y$ is
    $$
    \frac{\delta  e^{(1+\delta) y}}{\delta e^{(1+\delta) y} + (s/p) (1+\delta)  (1 - (j/p)^{\delta})}
    $$
    which is positive. \qedhere
  \end{enumerate}
  \end{proof}

\begin{lemma}
  \label{main-prop1}
  Suppose that $G$ has $s$ good cuts and $n$ vertices.  Then, for the Contraction Process for $J$ up to some stage $i$ with $\beta \leq i \leq n$, we have $\bE[ R_i ] \leq f(i - \beta,s,n - \beta)$.
\end{lemma}
\begin{proof}
We show this by induction on $n$. We will write $p = n - \beta, j = i - \beta$ and $m = |E(G)|$. The case $n = i$ is clear, since $R_i = 0 = f(i - \beta,s,i - \beta)$.  

For the induction step with $n > i$, the Contraction Process first selects an edge of $E(G) \setminus J$, arriving at a new graph $G'$ with $n-1$ vertices. So $$
\bE[R_i] = \frac{\lk}{m} + \bE[R_i^{G'}],
$$
where  $R_i^{G'}$ denotes the random variables defined in Eq.~(\ref{eq:Ri}) for graph $G'$. 

Let random variable $S'$ denote the number of good cuts in $G'$. By the induction hypothesis applied to $G'$, we have
\begin{equation}
\label{rig-eqn}
\bE[R_i] \leq (\lk/m) + \bE[ f(j,S',p-1)].
\end{equation}
Each good cut $C$ is selected with probability at least $\frac{|\partial C \setminus J|}{|E(G) \setminus J|} \geq \frac{k}{k-1} (1-\eps) \lk/m = (1+\delta) \lk/m$, so
$$
\bE[S'] \leq s (1 - (1 + \delta) \lk / m) \leq s e^{-(1+\delta) \lk / m }.
$$
By \Cref{gb1}, Jensen's inequality applies for the random variable $S'$ in Eq.~(\ref{rig-eqn}), giving:
$$
\bE[R_i] \leq (\lk/m) + f(j, \bE[S'], p-1) \leq (\lk/m) + f(j,s e^{-(1+\delta) \lk / m} ,p-1).
$$
Next, by \Cref{m-bound2},  we have $m \geq p \lk - \min\{s, p \} \lk/2$. So $\lk/m \leq z$ where we define 
$$
z =  \frac{2}{2 p - \min\{s, p\} }
$$
Since $y + f(j, s e^{-(1+\delta) y}, p-1)$ is an increasing function of $y$, we therefore have
$$
\bE[R_i] \leq z + f(j, s e^{-(1+\delta) z} ,p-1).
$$
To finish the proof and complete the induction, it suffices to show $z + f(j, s e^{-(1+\delta) z}, p-1) \leq f(j,s,p)$,  or equivalently,
\begin{equation}
\label{g0}
e^{(1+\delta) (z + f(j, s e^{-(1+\delta) z}, p-1))} -e^{(1+\delta) f(j,s,p)} \leq 0.
\end{equation}
After substituting in the formula for $f$, this expands to:
\begin{equation}
\label{g0g0}
( \tfrac{p-1}{j} )^{1 + \delta} \Bigl( e^{(1+\delta) z} + ( \tfrac{s}{p-1} ) (1 + 1/\delta)  (1 - (\tfrac{j}{p-1})^{\delta}) \Bigr) -  ( \tfrac{p}{j} )^{1+\delta} \Bigl( 1 +  ( \tfrac{s}{p} ) (1 + 1/\delta) (1 - (\tfrac{j}{p})^{\delta}) \Bigr) \leq 0.
\end{equation}
To simplify further, let us define a number of terms:
$$
r = s/p, \qquad q = j/p, \qquad \theta = 1 - 1/p, \qquad t = 2 - \min \{ r, 1 \}.
$$
We thus have $\tfrac{p-1}{j} = \theta/q,  \tfrac{s}{p-1} = r/\theta$, and $z = 2(1-\theta)/t$. The inequality in Eq.~(\ref{g0g0}) becomes:
$$
( \theta/q )^{1 + \delta} \Bigl( e^{2 (1-\theta) (1+\delta)/t} + (r/\theta ) (1 + 1/\delta)  (1 - (q/\theta)^{\delta}) \Bigr) -  ( 1/q )^{1+\delta} \Bigl( 1 +  r (1 + 1/\delta) (1 - q^{\delta}) \Bigr) \leq 0.
$$
Clearing out common factor $q^{1+\delta}$ and multiplying the left term through by $\theta^{1+\delta}$, it is equivalent to:
$$
 \bigl( \theta^{1+\delta} e^{2 (1-\theta)(1+\delta)/t} + r (1 + 1/\delta)  (\theta^{\delta} - q^{\delta}) \bigr) -  \bigl( 1 +  r (1 + 1/\delta) (1 - q^{\delta}) \bigr) \leq 0
$$

Collecting terms, multiplying through by $\delta$, and changing signs for convenience, Eq.~(\ref{g0}) is thus equivalent to showing:
\begin{equation}
\label{g1}
\delta + r (1 +\delta) (1 - \theta^{\delta}) - \delta \theta^{1+\delta} e^{2 (1-\theta) (1+\delta)/t} \geq 0.
\end{equation}

Note that parameter $q$ no longer plays a role in Eq.~(\ref{g1}).  Since $r \geq 2 -t$, it suffices to show that
\begin{equation}
  \label{g3}
\delta + (2-t) (1+\delta)  (1 - \theta^{\delta}) - \delta  \theta^{1+\delta} e^{2 (1-\theta) (1+\delta)/t} \geq 0.
\end{equation}

To show Eq.~(\ref{g3}), let us define a function 
$$
F(\theta, t) =  \delta +(2-t)  (1+\delta) (1 - \theta^{\delta}) - \delta  \theta^{1+\delta} e^{2 (1-\theta) (1+\delta)/t}
$$
for independent variables $\theta, t$. We need to show that $F(\theta, t) \geq 0$ for all $\theta  \in [0,1]$ and $t \in [1,2]$.

The second partial derivative of $F$ with respect to $t$ is given by
$$
\frac{\partial^2 F(\theta, t)}{\partial t^2} =\frac{-4 \delta  (1 + \delta) (1 - \theta) \theta ^{1 + \delta} e^{2 (1+\delta) (1-\theta)/t} ( (1+\delta ) (1 - \theta) + t)}{t^4}
$$
which is clearly negative for $\delta, \theta, t$ in the given range.  Thus,  the minimum value of $F(\theta, t)$ in the region occurs at either $t = 1$ or $t = 2$. So, in order to show that $F(\theta, t) \geq 0$, it suffices to show that $F(\theta, 1) \geq 0$ and $F(\theta, 2) \geq 0$.

At $t = 2$ we have $F(\theta, 2) = \delta(1 - e^{(1+\delta)(1-\theta)} \theta^{1+\delta})$. To show that $F(\theta, 2) \geq 0$, we thus need to show that $e^{(1+\delta)(1-\theta)} \theta^{(1+\delta)} \leq 1$, or equivalently $e^{1-\theta} \theta\leq 1$; this can be verified by routine calculus.

At $t = 1$, we have  $F(\theta,1) =  \delta + (1 + \delta)(1 - \theta^{\delta}) - \delta  \theta^{1+\delta} e^{2 (1-\theta) (1+\delta)}.$ Note that $F(1,1) = 0$. So, in order to show that $F(\theta, 1) \geq
0$ for all $\theta \in [0,1]$,  it suffices to show that the derivative of $F(\theta, 1)$ with respect to $\theta$ is negative for $\theta \in (0,1)$. This derivative is given by:
$$
\frac{\partial F(\theta, 1)}{\partial \theta} = -\delta  (1+\delta) e^{2 (1+\delta) (1-\theta)} \theta ^{-(1-\delta)} \bigl (e^{-2
   (1+\delta) (1-\theta)}-2 \theta ^2+\theta \bigr).
$$

To show this is negative, it suffices to show that $e^{-2 (1+\delta) (1-\theta)}-2 \theta ^2+\theta > 0$. Since $\delta \leq 1/2$, it suffices to show that $e^{-3 (1-\theta)}-2 \theta ^2+\theta > 0$ which can be verified by routine calculus for $\theta \in (0,1)$. 
This shows that $\frac{\partial F(\theta, 1)}{\partial \theta} \leq 0$, and so $F(\theta, t) \geq F(1, 1) = 0$.

Thus $F(\theta, t) \geq 0$ for all $\theta  \in [0,1]$ and $t \in [1,2]$ and hence the inequality of Eq.~(\ref{g3}) holds.
\end{proof}

\section{Putting it together: Bounds on the Contraction Algorithm}
\label{sec:final}
We now finish by getting our main bound for the Contraction Algorithm. 
\begin{lemma}
\label{llem1}
Suppose that $J$ is an edge set with $\alpha = |J|/\lambda_k$ and $n \geq i \geq 8 \alpha k^2 + 2 k$. Then $J$ survives the Contraction Algorithm to stage $i$ with probability at least $(n/i)^{-\alpha k} k^{-O(\alpha k^2)}$.
   \end{lemma}
   \begin{proof}
  Let us set $\eps = \frac{k+1}{2 k^2}$, and also define $\beta = k + 2 \alpha k/\eps$, and $j = i - \beta, p = n - \beta$ and $\delta = \frac{1}{2 k}$. By \Cref{main-extremal-theorem}, the number of medium cuts in $G$ is at most $a n$ for $a = k^{O(k)}$, and so \Cref{main-prop1} gives:
  $$
  \bE[R_i] \leq f(j, a n, p) = \log(p/j)  + \frac{\log \bigl( 1 + \frac{a n}{p} (1 + 1/\delta) (1 - (j/p)^{\delta}) \bigr)}{1 + \delta}.
  $$
  
Our condition on $i$ ensures $i \geq 2 \beta$. So $p \geq n/2$ and $j \geq i/2$, and thus $\log (p/j) \leq \log(n/i) + O(1)$ and $ a n / p \leq 2 a$. We therefore have
  $$
  \bE[R_i] \leq \log(n/i) + \frac{\log \bigl( 1 + 2 a (1 + 1/\delta) \bigr)}{1+\delta}   + O(1)  \leq \log (n/i) + \log a + O(1)
  $$
 Note that $i \geq \max\{ 4 \alpha k, k \}$  as  required in \Cref{hf1}. Thus, $J$ survives with probability at least $ (n/i)^{-\alpha k} e^{ -\alpha k -  \alpha k( \log a + \log k +  O(1) ) }$. Since $a = k^{O(k)}$, this is at least  $(n/i)^{-\alpha k} k^{-O(\alpha k^2)}$.
  \end{proof}

\begin{theorem}
\label{main-thmA1}
Running the Contraction Algorithm with parameter $\tau = \lceil 20 \alpha k^2 \rceil$ produces any given $k$-cut $K$ of weight at most $\alpha \lambda_k$ with probability at least $n^{-\alpha k} k^{-O(\alpha k^2)}$.
\end{theorem}
\begin{proof}
If $n \geq \tau$, then \Cref{llem1} applied to $J = \partial K$ (noting that necessarily $\alpha \geq 1$) shows that $K$ survives to $G_{\tau}$ with probability at least $(n/\tau)^{-\alpha k} k^{-O(\alpha k^2)}$. Then $K$ is selected from $G_\tau$ with probability at least $k^{-\tau} \geq k^{-O(\alpha k^2)}$.  Combining these probability bounds gives the stated result. If $n < \tau$, then the Contraction Algorithm simply selects a random $k$-cut, and so $K$ is chosen with probability at least $k^{-n}  \geq k^{-O(\alpha k^2)}$.
\end{proof}

\begin{corollary}
There are at most $n^{\alpha k} k^{O( \alpha k^2)}$ many $k$-cuts in $G$ with weight at most $\alpha \lambda_k$.
\end{corollary}

We could enumerate these $k$-cuts by repeatedly running the Contraction Algorithm, but each iteration would cost $O(n^2)$ time giving an overall runtime of roughly $O(n^{\alpha k + 2})$.  The next result shows how to remove this extraneous $n^2$ factor using a recursive version of the Contraction Algorithm from \cite{KS96}.  Note that directly printing out the $k$-cuts could take $\Omega(n^{\alpha k + 1})$ time, since each $k$-cut defines a partition of $V$. Hence, the algorithm necessarily produces the collection of $k$-cuts in a compressed data structure, which supports basic operations such as counting, sampling, etc. See \cite{KS96} or \cite{HarrisSrinivasan} for a more in-depth discussion.

\begin{theorem}
\label{rca-thm}
For each $k\ge3$, there is an algorithm to enumerate all $k$-cuts of weight at most $\alpha \lambda_k$ in time $n^{\al k}(\log n)^{O(\al k^2)}$ with probability at least $1 - 1/\poly(n)$.
\end{theorem}
\begin{proof}
First, if $n \leq 2^k$, then we directly use the Contraction Algorithm to stage $\tau = \lceil 4 \alpha k \rceil$. By \Cref{simpl-cor1}, this enumerates any given $k$-cut with probability at least $n^{-2 \alpha(k-1)} k^{-O(\alpha k)}$, so we must run it for $n^{2 \alpha (k-1)} k^{O(\alpha k)} \cdot \poly(\alpha, k, \log n)$ trials to get them all. This gives overall runtime of  $n^{2 \alpha (k-1)} k^{O(\alpha k)} \cdot \poly(\alpha, k, \log n) \cdot O(n^2)$, which is at most $e^{O( \alpha k^2)}$ by our assumption on $n$.  We thus assume for the remainder of the proof that $n \geq 2^k$.

We use a recursive algorithm, whose state is represented as a pair $(H,\ell)$ where $H$ is the current graph and $\ell = 0, \dots, T$ is the current level in the recursion.   The algorithm begins with the input graph $(G, 0)$ at level $\ell = 0$.  Given input $(G_\el,\el)$ at level $\ell$, there are two cases. If $\ell < T$, the algorithm runs $t_{\ell} = \lc (n_\el/n_{\el+1})^{\al k}\rc$ independent trials of the Contraction Algorithm to $n_{\el+1}$ vertices and recursively calls $(H,\el+1)$ for each resulting contracted graph $H$. Otherwise, if $\ell = T$, the algorithm outputs a randomly chosen $k$-cut. Here, the parameters $n_i$ are given by
 $$
n_i= \lc  \max\{  n^{ \bigl( \tfrac{2}{\al k} \bigr)^i } ,   20 \alpha k^2 \} \rc
$$
and the recursion depth $T$ is the first value with $n_T = \lceil 20 \alpha k^2 \rceil$.  Since $\log(20 \alpha k^2) \leq O( \alpha k)$ and $\alpha k / 2 \geq 3 / 2$, we have $T \leq O( \frac{\log \log n}{\log(\alpha k)})$.

To calculate the algorithm's success probability, fix some $k$-cut $K$ of $G$ with $|\partial K| \leq \alpha \lk$, and define a state $(G_\el,\el)$ to be \emph{successful} if no edge in $K$ has been contracted so far from $G_0 = G$ to $G_\el$. Clearly, $(G_0,0)$ is successful. For each successful input $(G_\el,\el)$ with $\el<T$, by \Cref{llem1} with $i=n_{\el+1}$, the probability that $K$ survives on each trial is at least
$ (n_\el/n_{\el+1})^{-\alpha k} \psi$ where $\psi = k^{-O(\alpha k^2)}.$ Over all $t_{\ell}$ trials, $K$ survives at least once with probability
$$
1 - \Bigl( 1 -  (n_\el/n_{\el+1})^{-\alpha k} \psi \Bigr)^{t_{\ell}}   \ge 1-e^{-\psi}  .
$$

Thus, given that some instance $(G_\el,\el)$ in the recursion tree is successful, the probability that at least one instance $(G_{\el+1},\el+1)$ is successful is at least $1 -e^{-\psi} \geq \psi/2$. Over all the $T=O(\frac{\log\log n}{\log (\alpha k)})$ levels of the recursion, the probability that there is some successful instance $(G_T,T)$ is at least $(\psi/2)^{T} \geq (\log n)^{-O(\alpha k^2)}$.  Finally the probability of selecting $K$ from a successful instance $(G_T, T)$ is at least $k^{-n_T} \geq k^{-O(\alpha k^2)}$. So, overall, $K$ is selected with probability at least 
$(\log n)^{-O(\alpha k^2)} \cdot k^{-O(\alpha k^2)}$. Due to our assumption that $n \geq 2^k$, this is at least $(\log n)^{-O(\alpha k^2)}$.

We now bound the runtime. For each level $\el < T$,  there are  $\prod_{j=0}^{\el-1} t_j \le \prod_{j=0}^{\el-1} 2 (n_j/n_{j+1})^{\al k} = 2^{\ell}(n_0/n_\el)^{\al k}$ instances $(G_\el,\el)$. In each such instance, the algorithm runs $t_{\ell}$ trials of the Contraction Algorithm, each taking $O(n_\el^2)$ time. The running time over all instances $(G_\el,\el)$ is therefore at most 
\[ 2^\el(n_0/n_\el)^{\al k}\cdot  t_{\ell}  \cdot O(n_\el^2) \le O \bigl( 2^T \cdot  (n_0/n_{\el+1})^{\al k}  \cdot n_\el^2 \bigr), \]
which is at most $O(2^T (2n_0)^{\al k})$ since $n_{\el+1} \ge n_\el^{2/(\al k)}/2$.  Summed over all $T$ recursion levels, the total runtime is at most $T \cdot O(2^T (2n_0)^{\al k}) \leq (2n)^{\alpha k} \cdot (\log n)^{O(1)}$.

If we repeat the entire recursive algorithm from $(G_0,0)$ a total of $(\log n)^{\Omega(\al k^2)} $ times, then $K$ is selected with probability at least $1/2$. There are $k^{O(\alpha k^2)} n^{\alpha k}$ many such $k$-cuts, so we run a further $O(\alpha k^2 \log k \log n)$ many trials to enumerate them all with probability $1 - 1/\poly(n)$. 
\end{proof}

As one concrete application, we get the main result:
\begin{theorem}
There is an algorithm to compute $\lambda_k$ in time $n^k(\log n)^{O(k^2)}$ for any value $k$.
\end{theorem}
\begin{proof}
For $k = 2$, this is the standard Recursive Contraction Algorithm of \cite{KS96}. Otherwise, apply \Cref{rca-thm} with $\alpha = 1$. This gives a large collection of $k$-cuts, which includes all the minimum $k$-cuts with high probability. We output the minimum weight of all $k$-cuts found. (The operation of taking minimum weight can be performed on the corresponding data structure.)
\end{proof}

%% file: appendix.tex
\section{Heuristic Bound on $R_i$}
\label{sec:heuristic-bound}
Given a graph $G$ with $s$ medium cuts, consider running
the Contraction Process for some edge set $J$ up to stage $i$. We will focus on the case where $s \ll n$; as it will later turn out, the resulting formulas are also correct (although not optimized) when $s$ is larger than $n$.

In each stage $j \geq i$, where the intermediate graph $G_j$ has $j$ vertices and $m_j$ edges, each good cut $C$ gets selected with probability $\frac{ |\partial C \setminus J| }{ m_j - |J|} \geq (1-\eps) \frac{k}{k-1} \lk / m_j = (1+\delta) \lk/m_j$. Letting $S_i$ denote the number of surviving good cuts at stage $i$, we thus have:
$$
\bE[S_i] \leq s \prod_{j=i+1}^n \Bigl( 1 - \frac{ (1+\delta) \lk}{ m_j} \Bigr) \leq s e^{-\sum_{j=i+1}^n (1+\delta) \lk / m_i} = s e^{-(1+\delta) R_i}.
$$

Since this is just a heuristic derivation, we blur the distinction between $\bE[S_i]$ and $S_i$, and we suppose that $S_i$ itself also satisfies this bound, i.e. $S_i \leq s e^{-(1+\delta) R_i}$.

We have $R_{i-1} = \frac{\lk}{m_i} + R_{i}$. By \Cref{m-bound}, we have $m_i   \geq  S_i \lk/2 + (i - S_i - \beta) \lk$, so
\begin{equation}
\label{barmeqn}
\frac{\bar\lambda_k}{m_i} \leq \frac{1}{i - \beta - S_i/2}.
\end{equation}

In order to carry out the induction proof later, we will need our bound on $R_i$ to have a simple closed form with nice concavity properties. In order to achieve this, we will need to use an upper bound on the quantity $\frac{\bar\lambda_k}{m_i}$ which is a \emph{linear} function of $S_i$. As we have mentioned, in the relevant case, we have $s \leq n$, and in this case we will also have $S_i \leq i$. We can then upper-bound the RHS of Eq.~(\ref{barmeqn}) by its secant line from $S_i = 0$ to $S _i= i - \beta$, yielding
  $$
  \frac{\lk}{m_i} \leq \frac{1}{i - \beta} \Bigl (1 + \frac{S_i}{i-\beta} \Bigr);
  $$
 note that by \Cref{m-bound0}, this upper bound will also be valid in the case where $S_i \geq i - \beta$.
  
Again ignoring any distinctions between random variables and their expectations,  this implies
  $$
  R_{i-1} \leq R_i +  \frac{1}{i - \beta} \Bigl (1 + \frac{ s e^{-(1+\delta) R_i}}{i-\beta} \Bigr).
  $$
  If we define $g(x) = R_{x+\beta}$ and $p = n - \beta$, then this can be relaxed to a differential equation with $g'(i) \approx R_{i+\beta} - R_{i-1 + \beta}$ defined as follows:
  $$
  g'(x) = \frac{-1}{x} \Bigl (1 + \frac{ s e^{-(1+\delta) g(x)}}{x} \Bigr), \qquad g(p) = 0.
  $$
  The differential equation has a closed-form solution:
  $$
  g(x) = \log(p/x) + \frac{\log \bigl(  1 + (s/p) (1 + 1/\delta) (1 - (x/p)^{\delta})  \bigr)}{1 + \delta}.
  $$
 Note the similarity of function $g$ to the function $f$ from
Eq.~\eqref{eq:fun-f} defined in \Cref{sec:small-cuts}.

